\documentclass[12pt]{amsart}
\usepackage[usenames]{color}
\usepackage{amssymb}
\usepackage{amsmath}
\usepackage{amsthm}
\usepackage{amsfonts}
\usepackage{amscd}
\usepackage{graphicx}
\usepackage{bm}

\usepackage[backend=biber,style=alphabetic,sorting=nyt]{biblatex}

\usepackage{tensor}

\usepackage{hyperref}

\theoremstyle{plain}
\newtheorem{theorem}{Theorem}

\newtheorem{lemma}[theorem]{Lemma}
\newtheorem{proposition}[theorem]{Proposition}

\theoremstyle{definition}
\newtheorem{definition}[theorem]{Definition}

\theoremstyle{remark}
\newtheorem{remark}[theorem]{Remark}

\usepackage{indentfirst} 

\usepackage{bm}          
\usepackage{url}
\usepackage{enumitem}

\usepackage[dvipsnames]{xcolor}

\usepackage{MnSymbol} 

\usepackage{arcs}
\usepackage{tikz-cd}
\usepackage[titletoc,title]{appendix}

\usetikzlibrary{arrows}

\AtBeginDocument{}%

\newcommand{\corurl}{red}
\newcommand{\corcite}{ForestGreen}
\newcommand{\corlink}{blue}
\newcommand{\iprod}{\raisebox{8pt}{\scalebox{1}[-2.5]{$\neg\,$}}}

\newcommand{\wed}{\wedge\hspace*{-7.5pt}.\hspace*{3pt}}
\newcommand{\otp}{\underset{\,^\cdot}{\otimes}}

\numberwithin{equation}{section}  

\hypersetup{linktocpage,colorlinks,urlcolor=\corurl,citecolor=\corcite,linkcolor=\corlink,
pdftitle={Hidden diffeos in a background independent field theory},
pdfauthor={J.F. Barbero, B. D\'{\i}az, J. Margalef-Bentabol, E.J.S. Villase\~nor}}

\begin{document}

\thanks{This work has been supported by the Spanish Ministerio de Ciencia 
Innovación y Universidades-Agencia Estatal de  Investigación Grant No. AEI/PID2020–116567GB-C22 and PID2021-128970OA-I00. B. D. acknowledges support from SECIHTI (M\'exico)  No. 371778. }


\title[Hidden diffeos in a background independent field theory]{Hidden diffeos in the Hamiltonian formulation of a background independent field theory}
\author[J.F. Barbero G.]{J. Fernando Barbero G.}
\address{J. Fernando Barbero G., \href{https://ror.org/05rtchs68}{Instituto de Estructura de la Materia}, IEM-CSIC, Serrano 123, 28006 Madrid, Spain}

\author[B. D\'{\i}az]{Bogar D\'{\i}az}
\address{Bogar D\'{\i}az, \href{https://ror.org/03n6nwv02}{Universidad Politécnica de Madrid}, Group of Biometrics, Biosignals, Security and Smart Mobility (GB2S), E.T.S.I. Minas y Energía, C. de Ríos Rosas, 21, 28003, Madrid, Spain \\ \href{https://ror.org/01tmp8f25}{Universidad Nacional Aut\'onoma de M\'exico}, Departamento de F\'isica de Altas Energ\'ias, Instituto de Ciencias Nucleares, Apartado Postal 70-543, Ciudad de M\'exico, 04510, M\'exico}

\author[J. Margalef-Bentabol]{Juan Margalef-Bentabol}
\address{Juan Margalef-Bentabol, \href{https://ror.org/03ths8210}{Universidad Carlos III de Madrid}, Departamento de Matemáticas, Avenida de la Universidad, 30 (edificio Sabatini), 28911, Legan\'es (Madrid), España}

\author[E.J.S. Villase\~nor]{Eduardo J.S. Villase\~nor}
\address{Eduardo J.S. Villase\~nor, \href{https://ror.org/03ths8210}{Universidad Carlos III de Madrid}, Departamento de Matemáticas, Avenida de la Universidad, 30 (edificio Sabatini), 28911, Legan\'es (Madrid), España}

\begin{abstract}

    We analyze from a geometric perspective the Hamiltonian formulation of a recent modification of the Husain-Kuchař model where, while preserving the connection as a dynamical variable, the other field is restricted to be the exterior covariant derivative of a Lie algebra-valued function. We prove that 3-dimensional diffeomorphisms can be accommodated among the local gauge transformations of the model in addition to the internal gauge symmetries.
\end{abstract}

\maketitle




%
%
\section{Introduction}{\label{sec_intro}}

Among the different approaches to the quantization of general relativity (GR), those taking classical Hamiltonian methods as the starting point have received much attention in the last decades. One such approach, loop quantum gravity (LQG), has inspired much work and is perceived as one of the leading contenders in this quest. However, despite having provided several tantalizing insights and interesting ideas, it is fair to say that we are still far from the satisfactory completion of this program.

A useful way to gain intuition about the several standing obstacles on the way to completing LQG is to consider simplified models that isolate some of its most relevant features but lack others. One of the main perceived difficulties in any quantum gravity program is dealing with diffeomorphisms. It is thus interesting to study diff-invariant field theories with dynamics simpler than the one corresponding to general relativity. A pioneering step in this direction is the so-called HK model proposed by V. Husain and K. Kuchař in \cite{HK}, where the authors introduced a field theory with some remarkable features. Its Hamiltonian description is similar to the Ashtekar formulation \cite{Ashtekar1, Ashtekar2} --both share the same phase space-- although with a crucial difference: the scalar constraint responsible for the ``non-trivial'' dynamics of GR is not present. This means that the HK model can be used as a testbed to study the quantization of diff-invariant theories in the Ashtekar phase space in a setting simpler than that of full general relativity.

Currently, several models related to the original HK proposal are known \cite{H,nosT,nosT2}. A recent and intriguing addition to this family of theories can be found in \cite{HM}. As we discuss later, this latest model is very similar to the original one but presents some surprising features, particularly regarding the constraints in the Hamiltonian formulation, which have been originally obtained by relying on the traditional Dirac ``algorithm''. The most striking of them is the apparent lack of a generator of local spatial diffeomorphisms. Here, we will study the same model using the geometric Hamiltonian methods proposed in \cite{GNH1,GNH2,GNH3,margalef2018towards,BPV}. We will devote special attention to obtaining the concrete form of the Hamiltonian vector fields. As we will show in detail, 3-dimensional diffeomorphisms \textit{can be found} among the local gauge transformations of the model despite the fact that they appear to be hidden \cite{HM}. 

There is an important comment that we want to make. Throughout this work, we will gloss over functional analytic aspects, which are not essential for our main result on the symmetry under diffeomorphisms. What we do here is, essentially, to rewrite the field equations in a way that lends itself easily to the interpretation of the dynamics of the model. Functional analytic issues would play a central role regarding the existence, uniqueness, and continuous dependence of the solutions to the field equations on initial data, issues that we do not care about in the present work. 

The paper is organized as follows. After this introduction, we devote Section \ref{sec_preliminary} to some preliminary material and explain our notation. Afterwards, in Section \ref{sec_HM_model}, we present the model considered here. We study its Hamiltonian formulation in Section \ref{sec_hamiltonian}. As the problem of determining the precise form of the Hamiltonian vector fields is central to the paper, we devote Section \ref{sec_ham_vector_fields} in full to this issue. We end with our conclusions.

%
%
\section{Preliminary results and notation}{\label{sec_preliminary}}
In this paper, we consider trivial principal (smooth) \( G \)-bundles over smooth manifolds \( M \), denoted as \( G \curvearrowright P \overset{\pi}{\to} M \), where \( G \) is a semisimple Lie group, with Lie algebra \( \mathfrak{g} \), acting on $P$. The archetypal example \( G = SU(2) \) will be the focus of the last section of the paper. Since the bundle is trivial, there exists a global section \( s \) of \( P \overset{\pi}{\to} M \). We denote by \( sg \) the corresponding global section induced by a smooth map \( g: M \to G \), namely \(x\mapsto sg(x)=s(x)\cdot g(x)\), where $\cdot$ is the right action of \(G\) on \(P\). Since the group action is fiber-preserving  and free and transitive on each fiber, every global section can be written in this form and the principal bundle is (non-canonically)  diffeomorphic to the product bundle \(M\times G\).

Using a notation similar to that in \cite{Freed1995} to deal with $\mathfrak{g}$-valued differential forms, given a connection one-form \( \alpha \in \Omega^1(P, \mathfrak{g}) \), with curvature  
\[
F_\alpha := \mathrm{d}\alpha + \frac{1}{2} [\alpha \wed \alpha] \in \Omega^2(P, \mathfrak{g}),
\]  
(as explained in detail below, ${[\cdot \wed \!\cdot{}\!]}$ denotes that the forms are wedge-multiplied while the $\mathfrak{g}$ coefficients are Lie-Bracket-multiplied), we define  
\[
A := s^*\alpha \in \Omega^1(M, \mathfrak{g}), \quad F_A := s^*F_\alpha = \mathrm{d}A + \frac{1}{2} [A \wed A] \in \Omega^2(M, \mathfrak{g})\,,
\]  
by pulling-back through $s$ the corresponding fields from $P$ to $M$. If we now use \( sg \) instead of \( s \) to construct the corresponding fields on \( M \), we obtain  
\[
(sg)^*\alpha = \mathrm{Ad}(g^{-1})A + g^{-1} \mathrm{d} g, \quad (sg)^*F_\alpha = \mathrm{Ad}(g^{-1}) F_A,
\]  
where, as usual, \( g^{-1} dg \) denotes the pullback to \( M \) by \( g \) of the (left) Maurer-Cartan form on \( G \).   The mappings  
\[
A \mapsto \mathrm{Ad}(g^{-1})A + g^{-1} \mathrm{d} g, \quad F_A \mapsto \mathrm{Ad}(g^{-1})F_A
\]  
are known as the internal \( G \)-gauge transformations of \( A \) and \( F_A \), respectively.

To fix our conventions and notation, we give some (mostly standard) definitions and list several results that will be used in the main body of the paper. The Lie product on $\mathfrak{g}$ is denoted by $[\cdot,\!\cdot]$. We dress $\mathfrak{g}$ with an \(\mathrm{Ad}(G)\)-invariant symmetric bilinear form $\langle \cdot,\! \cdot \rangle$  that, in many cases, coincides with (a multiple of)  the Cartan-Killing form of \(\mathfrak{g}\):
\[
\langle u , v \rangle=\kappa(u,v)=\mathrm{trace}_\mathfrak{g}(\mathrm{ad}(u)\circ\mathrm{ad}(u))=\mathrm{trace}_\mathfrak{g}(w\mapsto [u,[v,w])\,.
\]
We  denote the space of $\mathfrak{g}$-valued $p$-forms as $\Omega^p(M,\mathfrak{g})$.   If we choose a linear basis $(b_i)$ of $\mathfrak{g}$,  a $\mathfrak{g}$-valued $p$-form $A$ can be expanded as $A=b_i A^i$ (Einstein summation notation)  with $A^i\in \Omega^p(M)$. The structure constants $\tensor{C}{^k_i_j}$ of the basis $(b_i)$ are defined by
\[
[b_i,b_j]=b_k\tensor{C}{^k_i_j}\,.
\]
We also write 
\[k_{ij}=\langle b_i,b_j\rangle. \]

\begin{definition}\label{ledge}
Let $A\in\Omega^a(M,\mathfrak{g})$, $B\in\Omega^b(M,\mathfrak{g})$ and $(b_i)$ a basis of $\mathfrak{g}$ so that $A=b_iA^i$ and $B=b_iB^i$. We define
\begin{align*}
    \langle A\wed B\rangle&:=A^i\wedge B^j\langle b_i,b_j\rangle=k_{ij}  A^i\wedge B^j\in \Omega^{a+b}(M)\,,\\
    [A\wed B]&:=A^i\wedge B^j [b_i,b_j]=  b_k\tensor{C}{^k _i_j}   A^i\wedge B^j \in \Omega^{a+b}(M, \mathfrak{g})\,.
\end{align*}
These definitions are independent of the choice of basis.  
\end{definition}

\begin{definition} Given $A\in\Omega^1(M,\mathfrak{g})$, we can define the covariant exterior differential $\mathrm{d}_A$ and the curvature $F_A$ as  
\begin{align}
    \mathrm{d}_A B &:=\mathrm{d}B+[A\wed B] \label{cov_diff} \\
    &= b_k \mathrm{d}B^k + [b_i,b_j]  A^i\wedge B^j \nonumber
    = b_k (\mathrm{d}B^k+\tensor{C}{^k_i_j}A^i\wedge B^j )\,,\quad B\in \Omega^b(M,\mathfrak{g}) \,,\nonumber \\
    F_A &:= \mathrm{d}A+\frac{1}{2}[A\wed A] \\
    &= b_k \mathrm{d}A^k+[b_i,b_j] \frac{1}{2}A^i\wedge A^j    \label{curvature_def}
    = b_k\Big(\mathrm{d}A^k+\frac{1}{2} \tensor{C}{^k_i_j} A^i\wedge A^j\Big)\,.   
\nonumber
\end{align}

\end{definition}

The following properties collected in a (direct) Lemma will be extensively used in the rest of the paper:

\begin{lemma}\label{primer_lema} Let $A\in\Omega^a(M,\mathfrak{g})$, $B\in\Omega^b(M,\mathfrak{g})$ and $C\in\Omega^c(M,\mathfrak{g})$, then 
\begin{itemize}
    \item[(P1)] Graded-antisymmetry: \[\displaystyle[ A\wed B]=-(-1)^{ab} [B\wed A]\,.\]
    \item[(P2)] Graded-Jacobi identity: \[ \displaystyle (-1)^{ac}[ [A \wed B]\wed C]+(-1)^{ba}[ [ B \wed C]\wed A]+(-1)^{cb}[ [ C \wed A]\wed B]=0 \,.\]  
    \item[(P3)] Graded-symmetry: \[\displaystyle\langle A\wed B\rangle = (-1)^{ab}\langle B\wed A\rangle\,.\]
    \item[(P4)] Invariance: \[\langle[A\wed B]\wed C\rangle=\langle A\wed [B\wed C]\rangle\,.\]
\end{itemize}

\end{lemma}

\begin{remark}
(P1) and (P2) are a consequence of $\Omega(M,\mathfrak{g})=\bigoplus_k \Omega^k(M,\mathfrak{g})$ being a graded Lie superalgebra. (P4) is a consequence of the invariance of $\langle\cdot,\!\cdot\rangle$ and it is very useful to move brackets inside the arguments of  $\langle\cdot,\!\cdot\rangle$.
\end{remark}

\noindent Some other useful and standard differential geometry results are listed in the following Lemma that we also give without proof.
\begin{lemma}
Let $A\in \Omega^1(M,\mathfrak{g})$, $B\in\Omega^b(M,\mathfrak{g})$, $C\in\Omega^c(M,\mathfrak{g})$, and $X\in\mathfrak{X}(M)$, then:
\begin{itemize}
    \item[(P5)]  \(X\iprod [B\wed C]=[(X\iprod B)\wed C]+(-1)^b[B\wed (X\iprod C)]\,,\)

\vspace*{-.5mm}
    
    \item[(P6)]  \(\mathrm{d}_A[B\wed C]=[\mathrm{d}_AB\wed C]+(-1)^b[B\wed \mathrm{d}_AC]\,,\phantom{\big{|}}\)
    \item[(P7)] \( \mathrm{d}\langle B\wed C\rangle=\langle \mathrm{d}_A B\wed C\rangle+(-1)^b\langle B\wed \mathrm{d}_AC\rangle \,,\phantom{\big{|}}\)
    \item[(P8)] \(\mathrm{d}_A^2 B=[F_A\wed B]\phantom{\big{|}}\,.\)
\end{itemize}
\end{lemma}
\begin{remark}
    If $X\in\mathfrak{X}(M)$ is a smooth vector field, $A\in\Omega^a(M,\mathfrak{g})$ can act on $X$ to give $X\iprod A\in \Omega^{a-1}(M,\mathfrak{g})$  through $X\iprod A=b_i (X\iprod A^i)$. Of course, the definition is independent of the basis. 
\end{remark}

%
%
\section{The Husain-Mehmood model}{\label{sec_HM_model}}

Let $\Sigma$ be a closed, orientable, 3-manifold (therefore parallelizable). Let us introduce the 4-manifold $M=\mathbb{R}\times\Sigma$ foliated by ``spatial sheets'' $\Sigma_\tau:=\{\tau\}\times\Sigma$ ($\tau\in\mathbb{R}$). In this setting one has a canonical evolution vector field $\partial_{\mathrm{t}}\in \mathfrak{X}(M )$, transverse to every $\Sigma_\tau$ and defined by the tangent vectors to the curves $c_\sigma:\mathbb{R}\rightarrow M:\tau\mapsto (\tau,\sigma)$. There is also a scalar field $t\in C^\infty(M )$ defined as $t:M \rightarrow \mathbb{R}:(\tau,\sigma)\mapsto \tau$ and such that ${\partial_{\mathrm{t}}}\iprod \mathrm{d} t=1$.

Let us consider the action
\begin{equation}\label{action_HM}
S_{\mathrm{HM}}(\Phi,\mathrm{A})=\frac{1}{2}\int_{[\tau_1,\tau_2]\times \Sigma}\langle[\mathrm{d}_{\mathrm{A}}\Phi\wed \mathrm{d}_{\mathrm{A}}\Phi]\wed \mathrm{F}_{\mathrm{A}}\rangle\,.
\end{equation}
Here, the basic fields are  $\Phi\in\Omega^0(M ,\mathfrak{g})$ and $\mathrm{A}\in\Omega^1(M ,\mathfrak{g})$,  with prescribed values on the hypersurfaces \(\Sigma_{\tau_1}\) and \(\Sigma_{\tau_2}\).  
As can be seen, \eqref{action_HM} is similar to the HK action \cite{HK}
\begin{equation}\label{action_HK}
S_{\mathrm{HK}}(e,\mathrm{A})=\frac{1}{2}\int_{[\tau_1,\tau_2]\times \Sigma}\langle[e\wed e]\wed \mathrm{F}_\mathrm{A}\rangle\,.
\end{equation}
where $e\in\Omega^1(M,\mathfrak{su}(2))$. In fact, \eqref{action_HK} can just be obtained from \eqref{action_HK} by replacing $e$ by $\mathrm{d}_\mathrm{A}\Phi$ (for $\mathfrak{g}=\mathfrak{su(2)}$ in its original version). Consequently, one would expect the dynamics of these two models to be closely related. This is given by the corresponding Euler-Lagrange equations.

\begin{proposition}
    The field equations corresponding to the action principle \eqref{action_HM} are:
    \begin{subequations}\label{field_eqs}
\begin{align}
&{\langle\cdot\wed [[F_\mathrm{A}\wed\Phi]\wed F_\mathrm{A}]  \rangle}=0\,,\label{field_eqs_1}\\
&{\langle\cdot\wed \mathrm{d}_\mathrm{A}[[F_\mathrm{A}\wed\Phi]\wed \Phi] \rangle}=0\,.\label{field_eqs_2}
\end{align}
\end{subequations}
\end{proposition}
\begin{proof}
By using the results listed in the preceding section--and the vanishing of the variations of the fields on \(\Sigma_{\tau_1}\) and \(\Sigma_{\tau_2}\)-- the variations of the action \eqref{action_HM} can be written in the form:
\begin{equation}\label{variation}
\delta S_{\mathrm{HM}}=\int_{[\tau_1,\tau_2]\times \Sigma }\big(-\langle \delta\Phi\wed [[F_\mathrm{A}\wed\Phi]\wed F_\mathrm{A}] \rangle+ \langle \delta \mathrm{A}\wed \mathrm{d}_\mathrm{A} [[ F_\mathrm{A}\wed \Phi ]\wed \Phi ] \rangle \big)\,.
\end{equation}
The field equations \eqref{field_eqs} can be immediately read from this expression.
\end{proof}

Notice that for a semisimple Lie algebra $\mathfrak{g}$, the Cartan-Killing form is non-degenerate and, when \(\langle \cdot,\!\cdot \rangle =\kappa(\cdot,\!\cdot)\),  the field equations \eqref{field_eqs} reduce to
\begin{subequations}
\begin{align}
& [[F_\mathrm{A}\wed \Phi ]\wed F_\mathrm{A}]=0\,,\label{field_eqs_semisimple_1}\\
& \mathrm{d}_\mathrm{A}[[F_\mathrm{A}\wed \Phi]\wed \Phi]=0\,.\label{field_eqs_semisimple_2}
\end{align}
\end{subequations}
Here it is important to mention that, despite the claims made in \cite{HM}, equation \eqref{field_eqs_semisimple_1} does not imply \([F_\mathrm{A}\wed\Phi]=0\).

A useful way to unravel the meaning of the field equations \eqref{field_eqs} is to look at the Hamiltonian formulation corresponding to the action \eqref{action_HM}. This is the purpose of the following sections.

%
%
\section{The Hamiltonian formulation}{\label{sec_hamiltonian}}

To get the Hamiltonian formulation of the model considered here, we start by computing the Lagrangian. A standard way to do it on the manifold $M =\mathbb{R}\times\Sigma$ is by writing the action as
\[
S_{\mathrm{HM}}=\int_{M }\mathcal{L}=\!\!\int_{\mathbb{R}\times\Sigma}\!\!\mathrm{d}t\wedge\partial_{\mathrm{t}}\iprod\mathcal{L}=\!\!\int_{\mathbb{R}}\!\!\mathrm{d}t\!\!\int_{\Sigma_t} \!\partial_{\mathrm{t}}\iprod\mathcal{L}
=\!\!\int_{\mathbb{R}}\!\!\mathrm{d}t\!\!\int_{\Sigma}\jmath_t^*\partial_{\mathrm{t}}\iprod\mathcal{L}\!=\!\!\int_{\mathbb{R}}\!\!\!L(t)\,\mathrm{d}t =\!\!\!\int_{\mathbb{R}}\!\!L\,,
\]
where $\jmath_\tau:\Sigma\rightarrow\Sigma_\tau:\sigma\mapsto(\tau,\sigma)$ and
\[
L:\mathbb{R}\rightarrow \mathbb{R}:\tau\mapsto\int_\Sigma \jmath_\tau^*\partial_{\mathrm{t}}\iprod\mathcal{L}\,.
\]
A straightforward computation gives
\begin{align*}
  L(\tau)  & =\int_\Sigma \Big(\langle [\jmath_\tau^*(\pounds_{\mathrm{t}}\Phi)\wed \jmath_\tau^*(\mathrm{d}_\mathrm{A}\Phi)]\wed \jmath_\tau^*F_{\mathrm{A}} \rangle +\langle [[(\jmath_\tau^*\partial_{\mathrm{t}}\iprod \mathrm{A})\wed \jmath_\tau^*\Phi]\wed\jmath_\tau^*(\mathrm{d}_\mathrm{A}\Phi)]\wed \jmath_\tau^* F_\mathrm{A}\rangle \\
   &\hspace*{10mm}+\frac{1}{2}\langle [ \jmath_\tau^*(\mathrm{d}_\mathrm{A}\Phi)  \wed \jmath_\tau^*(\mathrm{d}_\mathrm{A}\Phi)]\wed \jmath_\tau^*\pounds_{\mathrm{t}}\mathrm{A}\rangle-\frac{1}{2}\langle  [ \jmath_\tau^*(\mathrm{d}_\mathrm{A}\Phi) \wed \jmath_\tau^*(\mathrm{d}_\mathrm{A}\Phi)]\wed \jmath_\tau^*(\mathrm{d}_\mathrm{A}(\partial_{\mathrm{t}}\iprod \mathrm{A}) )\rangle\Big)\,,
\end{align*}
where $\pounds_{\mathrm{t}}$ denotes the Lie derivative along the vector field $\partial_{\mathrm{t}}$. Let us look at the different terms appearing in this expression:

\medskip

\begin{itemize}
    \item  A basic object is the curve in $\Omega^0(\Sigma,\mathfrak{g})$ defined by
\[
\phi:\mathbb{R}\rightarrow\Omega^0(\Sigma,\mathfrak{g}):\tau\mapsto\phi(\tau)=\jmath_\tau^* \Phi\,.
\]
In terms of $\phi$ we have $\jmath_\tau^*(\pounds_{\mathrm{t}}\Phi)=\dot{\phi}(\tau)$ and, also,
\begin{align*}
 \jmath_\tau^*(\mathrm{d}_\mathrm{A} \Phi)
& =\jmath_\tau^*(\mathrm{d}\Phi+[\mathrm{A}\wed \Phi])=\mathrm{d}(\jmath_\tau^*\Phi)+[\jmath_\tau^*{\mathrm{A}\wed \jmath_\tau^*\Phi}]\\
&=\mathrm{d}\phi(\tau)+[A(\tau)\wed\phi(\tau)] =\mathrm{d}_{A(\tau)}\phi(\tau)=: (\mathrm{d}_A\phi)(\tau) \,,  
\end{align*}
where we have introduced the curve \[A:\mathbb{R}\rightarrow\Omega^1(\Sigma, \mathfrak{g}):\tau\mapsto A(\tau)=\jmath_\tau^*\mathrm{A}\,,\] with \(\dot{A}(\tau)=\jmath_\tau^*(\pounds_{\mathrm{t}}\mathrm{A})\phantom{\Big)}\).  

\item   We also have
\begin{align*}
\jmath_\tau^* F_\mathrm{A}&=\jmath_\tau^*\big(\mathrm{d} \mathrm{A}+\frac{1}{2}[\mathrm{A}\wed \mathrm{A}]\big)=\mathrm{d}\jmath_\tau^*\mathrm{A}+\frac{1}{2}[\jmath_\tau^*\mathrm{A}\wed\jmath_\tau^*\mathrm{A}]\\
&=\mathrm{d}A(\tau)+\frac{1}{2}[A(\tau)\wed A(\tau)] =:F_{A(\tau)}\,,\\
\jmath_\tau^*(\mathrm{d}_\mathrm{A}(\partial_{\mathrm{t}}\iprod \mathrm{A}))&=\mathrm{d}a(\tau)+\frac{1}{2}[A(\tau)\wed a(\tau)]=\mathrm{d}_{A(\tau)}a(\tau)=:(\mathrm{d}_Aa)(\tau)\,,
\end{align*}
where $a:\mathbb{R}\rightarrow \Omega^0(\Sigma,\mathfrak{g}): \tau\mapsto a(\tau)= \jmath_\tau^*(\partial_{\mathrm{t}}\iprod \mathrm{A})$. 
\end{itemize}
Hence
\begin{align*}
  L(\tau)&=\int_\Sigma \Big(\langle [\dot{\phi}(\tau)\wed (\mathrm{d}_A\phi)(\tau)]\wed (F_A)(\tau)\rangle +\langle [[a(\tau)\wed \phi(\tau)]\wed (\mathrm{d}_A\phi)(\tau)]\wed (F_A)(\tau)\rangle \\
   &\hspace*{5mm}+\frac{1}{2}\langle [(\mathrm{d}_A\phi)(\tau) \wed (\mathrm{d}_A\phi)(\tau)]\wed\dot{A}(\tau)\rangle-\frac{1}{2}\langle  [(\mathrm{d}_A\phi)(\tau) \wed (\mathrm{d}_A\phi)(\tau)]\wed (\mathrm{d}_Aa)(\tau) \rangle\Big)\,.
\end{align*}
From this expression, we can identify the configuration space $Q$ of the model,  with typical points $(\phi,A,a)\in Q$ :
\[
Q=\Omega^0(\Sigma,\mathfrak{g})\times\Omega^1(\Sigma,\mathfrak{g})\times\Omega^0(\Sigma,\mathfrak{g})
\]
and its (trivial) tangent bundle $TQ\cong Q\times Q$ with elements $\big((\phi,A,a),(v_\phi,v_A,v_{a})\big)$, that we will sometimes denote with the shorthand notation $v_q=(q,v)$. $L(t)$ can be obtained from the Lagrangian $\mathsf{L}:TQ\rightarrow\mathbb{R}$ given by
\begin{align}\label{Lagrangian1}
\begin{split}
  \mathsf{L}(v_q)&=\int_\Sigma \Big(\langle [v_\phi\wed \mathrm{d}_A\phi]\wed F_A\rangle +\langle [[a\wed \phi]\wed \mathrm{d}_A\phi]\wed F_A\rangle \\
   &\hspace*{42mm}+\frac{1}{2}\langle [\mathrm{d}_A\phi \wed \mathrm{d}_A\phi]\wed v_A\rangle-\frac{1}{2}\langle  [\mathrm{d}_A\phi \wed \mathrm{d}_A\phi]\wed \mathrm{d}_A a \rangle\Big)
\end{split}
\end{align}
by introducing a suitable space of curves $\mathcal{P}$ in $TQ$ parametrized by a time variable $t$ as $L(t)=\mathsf{L}\big(q(t),\dot{q}(t)\big)$. Some direct manipulations allow us to write \eqref{Lagrangian1} as follows: 
\begin{proposition}
    \label{thm_Lagrangian}
The Lagrangian \eqref{Lagrangian1} defined by the action \eqref{action_HM} is
\begin{equation}\label{Lagrangian2}
 \mathsf{L}(v_q)=\int_\Sigma \Big(\langle [v_\phi\wed \mathrm{d}_A\phi]\wed F_A\rangle +\frac{1}{2}\langle v_A\wed[\mathrm{d}_A\phi \wed \mathrm{d}_A\phi]\rangle+\langle a\wed \mathrm{d}_A[[F_A\wed\phi]\wed\phi]\rangle \Big)\,.
\end{equation}
\end{proposition}
\begin{proof}
Indeed, we have
\begin{align*}
\langle [[ a\wed \phi]\wed \mathrm{d}_A\phi]\wed F_A \rangle&=\langle [ a\wed\phi ] \wed [ \mathrm{d}_A\phi\wed F_A]\rangle = \langle  a\wed[\phi \wed [\mathrm{d}_A \phi\wed F_A]]\rangle\\
&=-\langle a\wed[[\mathrm{d}_A\phi\wed F_A]\wed\phi] \rangle=\langle a\wed[[F_A\wed \mathrm{d}_A\phi]\wed\phi] \rangle
\end{align*}
Also
\begin{align*}
0=\int_\Sigma -\frac{1}{2}\mathrm{d}_A\langle [\mathrm{d}_A\phi\wed \mathrm{d}_A\phi]\wed a   \rangle&=\int_\Sigma \Big( -\langle[[F_A\wed \phi]\wed \mathrm{d}_A\phi ]\wed a \rangle-\frac{1}{2}\langle[\mathrm{d}_A\phi\wed \mathrm{d}_A\phi]\wed \mathrm{d}_A a\rangle  \Big)
\end{align*}
so that
\[
\int_\Sigma -\frac{1}{2}\langle[\mathrm{d}_A\phi\wed \mathrm{d}_A\phi]\wed \mathrm{d}_A a\rangle=\int_\Sigma\langle a\wed[[F_A\wed \phi]\wed \mathrm{d}_A\phi]\rangle\,.
\]
Finally, we conclude that the two velocity-independent terms in \eqref{Lagrangian1} add up to
\begin{align*}
&\int_\Sigma \Big(\langle [[a\wed \phi]\wed \mathrm{d}_A\phi]\wed F_A\rangle-\frac{1}{2}\langle  [\mathrm{d}_A\phi \wed \mathrm{d}_A\phi]\wed \mathrm{d}_A a \rangle\Big)\\
=&\int_\Sigma\Big(\langle a\wed[[F_A\wed \mathrm{d}_A\phi]\wed \phi]\rangle+\langle a\wed[[F_A\wed\phi]\wed \mathrm{d}_A\phi]  \rangle\Big)\\
=&\int_\Sigma\Big(\langle a\wed[\mathrm{d}_A[F_A\wed\phi]\wed \phi]\rangle+\langle a\wed[[F_A\wed\phi]\wed \mathrm{d}_A\phi]  \rangle\Big)\\
=&\int_\Sigma \langle a\wed \mathrm{d}_A[[F_A\wed \phi]\wed \phi]\rangle\,.
\end{align*}    
\end{proof}

From this Lagrangian, we now get the Hamiltonian formulation for the present model. We start by computing the fiber derivative $\mathsf{FL}:TQ\rightarrow T^*Q, v_q\mapsto\mathsf{FL}(v_q)$,  of the Lagrangian defined in \eqref{Lagrangian2}
\begin{equation}
 w_q\mapsto 
\big(\mathsf{FL}(v_q)\big)(w_q)=\int_\Sigma\Big(\langle[w_\phi\wed \mathrm{d}_A\phi]\wed F_A\rangle+\frac{1}{2}\langle w_A\wed[\mathrm{d}_A\phi\wed \mathrm{d}_A\phi]\rangle\Big)\,,   \label{FD}
\end{equation}
where $w_q=(q,(w_\phi,w_A,w_a))$. Notice that \(\mathsf{FL}\) is neither injective nor surjective. If we denote the elements in the fiber $T_q^*\!Q$ as $\mathrm{p}_q:=(q,(\mathrm{p}_\phi,\mathrm{p}_A,\mathrm{p}_{a}))$,  \(\mathsf{FL}\) can be read as
\begin{align}
&\mathrm{p}_\phi(w_q)=\int_\Sigma\langle w_\phi\wed [\mathrm{d}_A\phi\wed F_A]\rangle\,, \nonumber\\
&\mathrm{p}_A(w_q)=\int_\Sigma\frac{1}{2}\langle w_A\wed[\mathrm{d}_A\phi\wed \mathrm{d}_A\phi]\rangle\,, \label{momenta}\\
&\mathrm{p}_{a}(w_q)\,=\,0\,.\phantom{\frac{1}{2}} \nonumber
\end{align}
These conditions define what in the standard physics parlance is known as the \textit{primary constraint submanifold} in phase space $\mathsf{M}_0:=\mathsf{FL}(TQ)$. This is the subset of $T^*Q$ where the dynamics takes place.  $\mathsf{M}_0$ is, actually, parameterized just by $(\phi,A,a)$. For this reason, the final formulation on $\mathsf{M}_0$ can be interpreted as defined directly on the configuration space $Q$. Once the fiber derivative is known, we can compute the energy 
\[
\mathsf{E}:TQ\rightarrow\mathbb{R}:v_q\mapsto \mathrm{p}_q(v_q)-\mathsf{L}(v_q)\,.
\]
In the present case, this is
\begin{equation}\label{energy}
\mathsf{E}(v_q)=\int_\Sigma \langle a\wed \mathrm{d}_A[\phi\wed[F_A\wed\phi]]\rangle \,.
\end{equation}
The Hamiltonian is defined implicitly by 
\[
\mathsf{H}\circ \mathsf{FL}=\mathsf{E}
\]
as a function \(\mathsf{H}:\mathsf{FL}(TQ)\rightarrow \mathbb{R}\).  
Notice that $\mathsf{H}$ is well defined because, from \eqref{FD},  \(\mathsf{FL}(u_q)=\mathsf{FL}(v_q)\) implies  \(\mathsf{E}(u_q)=\mathsf{E}(v_q)\).   The Hamiltonian is only defined on the image of the fiber derivative $\mathsf{FL}(TQ)$. As the energy is a function only of the fields defining the configuration space, the expression for the Hamiltonian is just \eqref{energy}
\begin{equation}\label{H}
 \mathsf{H}(\mathrm{p}_q)=\int_\Sigma \langle a\wed \mathrm{d}_A[\phi\wed[F_A\wed\phi]]\rangle\,.
\end{equation}
It depends on the momenta $\mathrm{p}_q$ only through their base point $q$.

Vector fields in the phase space considered here $\mathsf{Y}\in\mathfrak{X}(T^*\!Q)$ have the form
\begin{equation}\label{vector_field}
\mathsf{Y}=(\mathsf{Y}_\phi,\mathsf{Y}_A,\mathsf{Y}_{a},\mathsf{Y}_{\mathrm{p} \phi},\mathsf{Y}_{\mathrm{p} A},\mathsf{Y}_{\mathrm{p} a})\,,
\end{equation}
where
\begin{align*}
&\mathsf{Y}_\phi:T^*\!Q\rightarrow\Omega^0(\Sigma,\mathfrak{g})\,,\quad &&\mathsf{Y}_{\mathrm{p} \phi}:T^*\!Q\rightarrow \Omega^0(\Sigma,\mathfrak{g})^*\,,\\
&\mathsf{Y}_A:T^*\!Q\rightarrow\Omega^1(\Sigma,\mathfrak{g})\,,\quad &&\mathsf{Y}_{\mathrm{p} A}:T^*\!Q\rightarrow \Omega^1(\Sigma,\mathfrak{g})^*\,,\\
&\mathsf{Y}_a\,:\,T^*\!Q\rightarrow\Omega^0(\Sigma,\mathfrak{g})\,,\quad &&\mathsf{Y}_{\mathrm{p} a}\,:\,T^*\!Q\rightarrow \Omega^0(\Sigma,\mathfrak{g})^*\,,
\end{align*}
The $\mathsf{Y}_{\mathrm{p} \phi}$, $\mathsf{Y}_{\mathrm{p} A}$, $\mathsf{Y}_{\mathrm{p} a}$ are functions in $T^*\!Q$. They are ``dual'' in the sense that acting, respectively, on objects such as $\mathsf{Y}_\phi$, $\mathsf{Y}_A$, and $ \mathsf{Y}_{a}$ they give real functions in phase space. For instance, let $\mathrm{p}_q\in T^*\!Q$, then $\mathsf{X}_A(\mathrm{p}_q)\in\Omega^1(\Sigma,\mathfrak{g})$ and $\mathsf{Y}_{\mathrm{p} A}(\mathrm{p}_q)\big(\mathsf{X}_A(\mathrm{p}_q)\big)\in\mathbb{R}$. The result can be independent of $\mathrm{p}_q$ [i.e. ``constant'']. Notice that $\mathsf{Y}_{\mathrm{p} A}(\mathsf{X}_A(\cdot{}))$ is a real function in phase space.

Vector fields tangent to $\mathsf{M}_0$ have the form \eqref{vector_field} with
\begin{align}\label{tangent_vector_fields}
\begin{split}
  & \mathsf{Y}_{\mathrm{p} \phi}(\cdot)=\int_\Sigma\big(\langle[\cdot\wed \mathrm{d}_A \mathsf{Y}_\phi]\wed F_A\rangle+\langle[\cdot\wed[\mathsf{Y}_A\wed \phi]]\wed F_A\rangle+\langle[\cdot\wed \mathrm{d}_A\phi]\wed \mathrm{d}_A  \mathsf{Y}_A\rangle\big)\,, \\
  & \mathsf{Y}_{\mathrm{p} A}(\cdot)=\int_\Sigma \big(\langle\cdot\wed [\mathrm{d}_A \mathsf{Y}_\phi\wed \mathrm{d}_A \phi]\rangle+\langle\cdot\wed [[\mathsf{Y}_A\wed \phi]\wed \mathrm{d}_A\phi] \rangle\big)\,, \\
  & \mathsf{Y}_{\mathrm{p} a}=0\,.\phantom{\int}
\end{split}
\end{align}
The pullback of $\mathsf{d H}$ acting on a vector field $\mathsf{Y}_0\in\mathfrak{X}(\mathsf{M}_0)$ is
\begin{align*}
\begin{split}
 \mathsf{d H}(\mathsf{Y}_0)=\int_\Sigma&\Big(\langle \mathsf{Y}_\phi\wed\big([\mathrm{d}_A a\wed[F_A\wed\phi]]+[F_A\wed[\mathrm{d}_A a \wed\phi]]\big)\rangle\\
&\hspace*{-2mm}+\langle \mathsf{Y}_A\wed\big([[\phi\wed a]\wed[F_A\wed \phi]]-[\mathrm{d}_A\phi\wed[\mathrm{d}_A a \wed \phi]]\\
 &\hspace*{1.4cm}+[\phi\wed[[a\wed\phi]\wed F_A]]+[\phi\wed[\mathrm{d}_A a\wed \mathrm{d}_A\phi]]\big)\rangle\phantom{\int}\\
&\hspace*{-2mm}+\langle \mathsf{Y}_{a}\wed \mathrm{d}_A[\phi\wed[F_A\wed \phi]]\rangle\Big)\,.\phantom{\int}
\end{split}
\end{align*}
Notice that, as the Hamiltonian depends only on the base point, the pullback of $\mathsf{d H}$ onto $\mathsf{M}_0$ can be obtained without ever making use of \eqref{tangent_vector_fields}.

We now compute the pullback of the canonical symplectic form onto $\mathsf{M}_0$, which can be obtained from
\[
\Omega(\mathsf{X},\mathsf{Y})=\mathsf{Y}_{\mathrm{p} \phi}(\mathsf{X}_\phi)-\mathsf{X}_{\mathrm{p} \phi}(\mathsf{Y}_\phi)+\mathsf{Y}_{\mathrm{p}A}(\mathsf{X}_A)-\mathsf{X}_{\mathrm{p} A}(\mathsf{Y}_A)+\mathsf{Y}_{\mathrm{p} a}(\mathsf{X}_{a})-\mathsf{X}_{\mathrm{p} a}(\mathsf{Y}_{a})
\]
by making use of \eqref{tangent_vector_fields}. We get
\begin{align*}
  \hspace*{-5mm}\Omega(\mathsf{X}_0,\mathsf{Y}_0)=\int_\Sigma\Big(&\langle \mathsf{Y}_\phi\wed ([\mathsf{X}_A\wed[F_A\wed \phi]]+[F_A\wed[\mathsf{X}_A\wed \phi]])\rangle\\
  &\hspace*{-4mm}+\langle \mathsf{Y}_A\wed([\phi\wed[\mathsf{X}_A\wed \mathrm{d}_A\phi]]-[\mathrm{d}_A\phi\wed[\mathsf{X}_A\wed \phi]]\\
  &\hspace*{1.3cm}-[\phi\wed[\mathsf{X}_\phi\wed F_A]]+[\mathsf{X}_\phi\wed[F_A\wed \phi]])\rangle\Big)\,.
\end{align*}
Now we need to find what conditions we get from $ \mathsf{dH}(\mathsf{Y}_0)=\Omega(\mathsf{X}_0,\mathsf{Y}_0)$ for every $\mathsf{Y}_0$. These will depend on the properties of the invariant form $\langle \cdot,\!\cdot\rangle$. If $\langle \cdot,\!\cdot\rangle$ is degenerate, a condition of the type
\[
\int_\Sigma \langle \mathsf{Y}_\phi\wed Z\rangle=0
\]
for all $\mathsf{Y}_\phi$ will not imply $Z=0$ but only $\langle \cdot \wed Z\rangle=0$. Keeping this in mind, we find the \textit{secondary constraints}
\begin{equation}\label{constraint}
\langle\cdot\wed \mathcal{C}\rangle=0\,,\qquad\qquad \mathcal{C}:=\mathrm{d}_A[\phi\wed[F_A\wed\phi]]\,,
\end{equation}
which define the submanifold 
\begin{equation}
  \mathsf{M}_1:=\mathsf{M}_0\cap \{\mathrm{p}_q\in T^*\!Q\,:\, \langle\cdot\wed \mathcal{C}\rangle=0\}\,,  \label{M1}
\end{equation}
and the following conditions for the components of the Hamiltonian vector fields
\begin{subequations}\label{eqs:X}
\begin{align}
  & \langle\cdot\wed \Big\{[(\mathsf{X}_A-\mathrm{d}_Aa)\wed[F_A\wed\phi]]+[F_A\wed[(\mathsf{X}_A-\mathrm{d}_Aa)\wed\phi]]\Big\}\rangle=0 \label{eq1}\\
  & \langle\cdot\wed \Big\{[(\mathsf{X}_\phi-[\phi\wed a])\wed [F_A\wed\phi]]-[\mathrm{d}_A\phi\wed[(\mathsf{X}_A-\mathrm{d}_Aa)\wed\phi ]]\label{eq2}\\
  &\hspace*{1cm} -[\phi\wed[(\mathsf{X}_\phi-[\phi\wed a])\wed F_A]] +[\phi\wed[(\mathsf{X}_A-\mathrm{d}_Aa)\wed \mathrm{d}_A\phi]]\Big\}\rangle=0\,.\nonumber
\end{align}
\end{subequations}
Notice that they do not involve $\mathsf{X}_{a}$, so this component is arbitrary at this stage. The vector fields found by solving \eqref{eqs:X} must be tangent to the set defined by the constraints \eqref{constraint}. The tangency condition is
\begin{align}
  \label{tangency}\begin{split}
&\langle\cdot\wed \Big\{{\mathrm{d}_A[\mathsf{X}_\phi\wed[F_A\!\!\wed\phi]]}+\!{\mathrm{d}_A[\phi\wed [\mathrm{d}_A \mathsf{X}_A\!\wed \phi]]}\\
&\hspace{3.8cm}+{\mathrm{d}_A[\phi\wed[F_A \wed \mathsf{X}_\phi]]}+{[\mathsf{X}_A \!\wed [\phi\wed[F_A \wed\phi]]]}\rangle\Big\}=0\,. \end{split}
\end{align}
The elements of the Hamiltonian formulation for the action \eqref{action_HM} are summarized in the following proposition, whose proof has been spelled out in the preceding paragraphs
\begin{proposition}\label{th_Hamiltonian}
The dynamics of the model defined by the action \eqref{action_HM} takes place on $\mathsf{M}_1\subset T^*Q$ given by \eqref{M1}, which is defined as the intersection of the primary constraint subset $\mathsf{M}_0$ and the subset defined by the secondary constraints \eqref{constraint}. Furthermore, the components $(\mathsf{X}_\phi,\mathsf{X}_A, \mathsf{X}_{a})$ of the Hamiltonian vector fields restricted to $\mathsf{M}_1$ must satisfy the equations \eqref{eqs:X} and are also subject to the consistency conditions \eqref{tangency}.
\end{proposition}
\begin{remark}\label{remark_ps} When \(\langle \cdot,\!\cdot\rangle\) is non-degenerate, we can write  
\begin{align*}
&\mathrm{p}_\phi(\cdot)=\int_\Sigma\langle \cdot\wed p_\phi \rangle\,, \nonumber\\
&\mathrm{p}_A(\cdot)=\int_\Sigma \langle \cdot \wed p_A \rangle\,,\\
&\mathrm{p}_{a}(\cdot)\,=\,\int_\Sigma\langle \cdot\wed p_a \rangle\,. 
\end{align*}
From \eqref{momenta} we obtain \(p_\phi=[\mathrm{d}_A\phi\wed F_A]\,,\, p_A=\frac{1}{2}[\mathrm{d}_A\phi\wed \mathrm{d}_A\phi]\,, \, p_a=0\). Then, we have \(\mathsf{M_1}\subset\mathsf{M_0}\subset Q\times P\) where
\[
Q=\Omega^0(\Sigma,\mathfrak{g})\times\Omega^1(\Sigma,\mathfrak{g})\times\Omega^0(\Sigma,\mathfrak{g})\,,\quad 
P=\Omega^3(\Sigma,\mathfrak{g})\times\Omega^2(\Sigma,\mathfrak{g})\times\Omega^3(\Sigma,\mathfrak{g})\,, 
\]
and
\begin{align*}
 \mathsf{M_0}&=\{((\phi,A,a),(p_\phi,p_A,p_a))\in Q\times P\,:\, p_\phi=[\mathrm{d}_A\phi\wed F_A]\,,\, p_A=\frac{1}{2}[\mathrm{d}_A\phi\wed \mathrm{d}_A\phi]\,, \, p_a=0 \}\,, \\
  \mathsf{M_1}&=\{((\phi,A,a),(p_\phi,p_A,p_a))\in \mathsf{M_0}\,:\, \mathrm{d}_A[\phi\wed[F_A\wed\phi]]=0 \} \,.
\end{align*}
  
\end{remark}

Several comments are in order now. Equations \eqref{eqs:X} give the components of the Hamiltonian vector fields that define the evolution of the model. They are a central element of the Hamiltonian formulation. In the present case, they can be written as homogeneous equations in the variables
\begin{subequations}\label{eq:Z}
\begin{align}
&\mathsf{Z}_\phi:=\mathsf{X}_\phi-[\phi\wed a]\,,\\
&\mathsf{Z}_A:=\mathsf{X}_A-\mathrm{d}_Aa\,,
\end{align}
\end{subequations}
so these equations can always be solved without any extra constraints arising (recall also that $\mathsf{X}_a$ is arbitrary at this stage). Once solved, it is crucial to require that the tangency condition \eqref{tangency} holds because, in principle, new constraints or extra conditions on the Hamiltonian vector fields can originate from this consistency condition. In the present case, it is possible to discuss the tangency condition without explicitly knowing the form of the solutions to \eqref{eqs:X}. To this end, we plug \eqref{eq:Z} into  \eqref{eqs:X} to obtain $\langle\cdot \wed \mathcal{A}\rangle=0$ and $\langle\cdot\wed \mathcal{B}\rangle=0$ where
\begin{subequations}
\begin{align}
  & \mathcal{A}:= [\mathsf{Z}_A\wed[F_A\wed\phi]]+[F_A\wed[\mathsf{Z}_A\wed\phi]]\,,\label{eeq1}\\
  & \mathcal{B}:= [\mathsf{Z}_\phi\wed [F_A\wed\phi]]-[\mathrm{d}_A\phi\wed[\mathsf{Z}_A\wed\phi ]]-[\phi\wed[\mathsf{Z}_\phi\wed F_A]] +[\phi\wed[\mathsf{Z}_A\wed \mathrm{d}_A\phi]]\label{eeq2}\,.
\end{align}
\end{subequations}
A short computation using the results of Lemma \ref{primer_lema} proves that \eqref{tangency} is equivalent to 
\begin{equation}
    \langle\cdot\wed\Big(\mathrm{d}_A\mathcal{B}-[\phi\wed\mathcal{A}]+[a\wed\mathcal{C}]\Big)\rangle=0-\langle[\cdot\wed\phi]\wed \mathcal{A}\rangle+\langle[\cdot\wed a]\wed\mathcal{C}\rangle=0\,.
\end{equation}
\begin{remark}
    The constraint \eqref{constraint} can be obtained by pulling back the field equation \eqref{field_eqs_semisimple_2} to the slices of the foliation of $M $ defined by the function $t$. However, as the other equation \eqref{field_eqs_semisimple_1} is a top-form in $M$, its pullback vanishes.
\end{remark}

\begin{remark} 
In the case where \(\langle \cdot,\!\cdot\rangle\) is non-degenerate, if we define (see remark \ref{remark_ps})
\begin{align*}
    \mathcal{C}_\phi&:=p_\phi-[\mathrm{d}_A\phi\wed F_A]\,,\\
    \mathcal{C}_A&:=p_A-\frac{1}{2}[\mathrm{d}_A\phi\wed \mathrm{d}_A\phi]\,,
\end{align*}
then 
\[
\mathrm{d}_A\mathcal{C}_A+[\phi\wed \mathcal{C}_\phi]\overset{\eqref{constraint}}{=}\mathrm{d}_Ap_A+[\phi\wed p_\phi]+\mathcal{C}\,.
\]
Hence, on the constraint surface $\mathsf{M}_1$, we have
\[\mathrm{d}_Ap_A+[\phi\wed p_\phi]=0\,.\]
This is precisely the constraint appearing in \cite{HM} for $G=SU(2)$ written in terms of differential forms. It has the expected form of the generator of internal $G$-gauge transformations (the Gauss law).
\end{remark}

%
%
\section{Finding the Hamiltonian vector fields}{\label{sec_ham_vector_fields}}

To understand and disentangle the dynamics of the model considered here, it is necessary to solve \eqref{eqs:X}. Although this is a simple problem in principle --just a set of linear algebraic equations-- in practice, it requires some effort. An important observation to make at this point is that we are interested in the solutions to these equations when the secondary constraints \eqref{constraint} hold. Notice that if the determinant of the matrix that defines the linear system vanishes when the constraints hold, the system will have non-trivial solutions depending on some arbitrary objects defined on the phase space of the system. Hence, in addition to the arbitrary $a$, other arbitrary local gauge parameters may be present.

The resolution of \eqref{eqs:X} for general Lie algebras is beyond the scope of the present paper. Here, we will concentrate on a concrete (and physically relevant) example: the semisimple Lie algebra $\mathfrak{g}=\mathfrak{su}(2)$. In this case \[\langle \cdot,\!\cdot \rangle =-\frac{1}{2}\kappa(\cdot,\!\cdot) \] 
turning \(\mathfrak{su}(2)\) into a Euclidean vector space. The equations to be solved for $\mathsf{Z}_\phi$ and $\mathsf{Z}_A$ are
\begin{subequations}
\begin{align}
  & [\mathsf{Z}_A\wed[F_A\wed\phi]]+[F_A\wed[\mathsf{Z}_A\wed\phi]]=0\,, \label{eeeq1}\\
  & [\mathsf{Z}_\phi\wed [F_A\wed\phi]]-[\mathrm{d}_A\phi\wed[\mathsf{Z}_A\wed\phi ]]-[\phi\wed[\mathsf{Z}_\phi\wed F_A]] +[\phi\wed[\mathsf{Z}_A\wed \mathrm{d}_A\phi]]=0\label{eeeq2}\,.
\end{align}
\end{subequations}
subject to the condition that the secondary constraints \eqref{constraint}, which in this case reduce to $\mathrm{d}_A[\phi\wed[F_A\wed\phi]]=0$, hold.

We proceed by choosing a $\langle \cdot,\!\cdot \rangle$-orthonormal  basis \((b_i)\) of $\mathfrak{su}(2)$ such that the structure constants are $\tensor{C}{^k_i_j}=\tensor{\varepsilon}{^k_i_j}$, where $\tensor{\varepsilon}{^k_i_j}=\delta^{kl}\varepsilon_{lij}$ and $\varepsilon_{ijl}$ denotes the totally antisymmetric Levi-Civita symbol such that $\varepsilon_{123}=+1$. The Cartan-Killing form in this basis is $k_{ij}=\varepsilon_{ik}^{\phantom{ij}l}\varepsilon_{jl}^{\phantom{jl}k}=-2\delta_{ij}$. By working in the chosen Lie algebra basis and defining 
\[\tensor{\mathsf{P}}{_{i_1}_{i_2}^{j_1}^{j_2}}=2\delta_{i_1}^{j_1}\delta_{i_2}^{j_2}-\delta_{i_2}^{j_1}\delta_{i_1}^{j_2}-\delta_{{i_1}{i_2}}\delta^{{j_1}{j_2}},\] the equations \eqref{eqs:X} become
\begin{subequations}\label{eqs:Z with P}
\begin{align}
&\tensor{\mathsf{P}}{_i_l^j^k}{\mathsf Z}^i_{A}\wedge F_j\phi_k=0\,,\label{eq1_components}\\
&\tensor{\mathsf{P}}{_i_l^j^k}\big(\mathsf{Z}^i_{\phi}F_k-(\mathrm{d}_A\phi)^i\wedge \mathsf{Z}_{Ak}\big)\phi_j=0\,,\label{eq2_components}
\end{align}
\end{subequations}
and the constraint \eqref{constraint} is
\begin{equation}\label{constraint_components}
\tensor{\mathsf{P}}{_i_l^j^k}(\mathrm{d}_A\phi)^i\wedge F_k\,\phi_j=0\,.
\end{equation}

As our purpose is to understand the gauge symmetries of the model and, in particular, figure out what the gauge parameters are, we will solve the equations under the simplifying condition that the $\mathrm{d}_A\phi$ are locally (i.e., on an open subset $U$ of $\Sigma$) linearly independent. This is one of the main hypotheses used in \cite{HM} to analyze the dynamics of this system, where it is required that
\[\langle \mathrm{d}_A \phi \wed [ \mathrm{d}_A \phi \wed \mathrm{d}_A \phi]\rangle\in \Omega^3(\Sigma)\] is a volume form and also that
\[\langle \mathrm{d}_A \phi \otp \mathrm{d}_A \phi \rangle =\kappa_{ij} (\mathrm{d}_A \phi )^i\otimes (\mathrm{d}_A \phi )^j\] is a non-degenerate metric on $\Sigma$ (whereas $\langle \mathrm{d}_{\mathrm{A}} \Phi \otp \mathrm{d}_{\mathrm{A}} \Phi \rangle =\kappa_{ij} (\mathrm{d}_{\mathrm{A}} \Phi)^i\otimes (\mathrm{d}_{\mathrm{A}} \Phi )^j$ is degenerate on $M$).

Let us introduce the following local expansions for $F\in\Omega^2(U,\mathfrak{g})$ and $\mathsf{Z}_A\in\Omega^1(U,\mathfrak{g})$
\begin{subequations}
\begin{align}\label{expansions}
&F_i=\frac{1}{2}\mathbb{F}_{ij}\varepsilon^{jkl}(\mathrm{d}_A\phi)_k\wedge (\mathrm{d}_A\phi)_l\,,\\
&\mathsf{Z}_{Ai}=\mathbb{Z}_{ij}(\mathrm{d}_A\phi)^j\,.
\end{align}
\end{subequations}
For notational convenience, we also define
\begin{equation}\label{other_notations}
\mathbb{Z}:=\delta^{ij}\mathbb{Z}_{ij}\,,\quad\mathbb{A}_i:=\varepsilon_{ijk}\mathbb{Z}^{jk}\,,\quad\mathbb{L}_i:=(\phi.\mathbb{Z})_i\,,\quad \mathbb{F}:=\delta^{ij}\mathbb{F}_{ij}\,.
\end{equation}
where we denote $(A.v)_i:=A_{ij}v^j$ and $(v.A)_i:=v^jA_{ji}$ (adjacent indices contracted). With this choice, \eqref{eqs:Z with P} is equivalent to
\begin{subequations}
\begin{align}
  & \tensor{\mathsf{P}}{_i_l^j^k}\mathbb{Z}^{im}\mathbb{F}_{jm}\phi_k=0\,,\label{eeq1_components} \\
  & \tensor{\mathsf{P}}{_i_l^j^k}\big(\mathsf{Z}^i_{\phi}\mathbb{F}_{ks}-\tensor{\varepsilon}{^i^r_s}\mathbb{Z}_{kr}\big)\phi_j=0\,,\label{eeq2_components}
\end{align}
\end{subequations}
and the secondary constraints become
\begin{equation}\label{constraint_Dphi}
\mathbb{C}_i:=2(\mathbb{F}.\phi)_i-(\phi.\mathbb{F})_{i}- \mathbb{F}\phi_i=0\,.
\end{equation}
Our strategy is to, first, solve \eqref{eeq2_components} for $\mathbb{Z}_{ij}$ in terms of $\mathsf{Z}_{\phi i}$ and then find out if \eqref{eeq1_components} imposes any conditions on the $\mathsf{Z}_{\phi i}$ and/or any other free parameters that may show up in the solution to \eqref{eeq2_components}. Before proceeding, a comment is in order. Solving the equations can be done without much effort by using standard algebraic manipulation software. It is harder, though, to simplify and find a compact and transparent way to write the solutions so that their Lie-algebra ``index structure'' is apparent, something necessary to interpret and use them. 

For the purpose of solving the equations for the Hamiltonian vector fields, the following simple Lemma comes in handy
\begin{lemma}\label{lemma_eq}
Let $\phi,Q\in \mathfrak{su}(2)$, with $\phi\neq 0$. Then the equation
\begin{equation}\label{equation_z}
[\phi,z]=Q
\end{equation}
in the unknown $z\in\mathfrak{su}(2)$,  admits solutions if and only if $\phi.Q=0$. In this case, the solutions are given by
\begin{equation}\label{solution_z}
z=-\frac{1}{\phi^2}[\phi,Q]+N\phi\,,
\end{equation}
with $N\in\mathbb{R}$ arbitrary and $\phi^2:=\phi.\phi$.
\end{lemma}
\begin{proof}
That $\phi.Q=0$ is necessary for the compatibility is evident using (P4) of Lemma \ref{primer_lema}. To see that it is also sufficient, we make use of $[\alpha,[\beta,\gamma]]= (\alpha.\gamma) \beta -(\alpha.\beta)\gamma$ to show that a particular solution to \eqref{equation_z} is\[
 z=-\frac{1}{\phi^2}[\phi,Q]\,.
\]
On the other hand, in $\mathfrak{su}(2)$, 
\[
[\phi,z]=0\Leftrightarrow z=N \phi\,.
\]
\end{proof}

Another important result that we will use is the following:
\begin{lemma}
Let $\displaystyle \mathbb{Q}_{ls}:= \tensor{\mathsf{P}}{_i_l^j^k}\mathsf{Z}^i_{\phi}\mathbb{F}_{ks}\phi_j$ and $\mathbb{Q}:=\delta^{ij}\mathbb{Q}_{ij}$, then the constraints \eqref{constraint_Dphi} imply
\begin{equation}\label{condition_Q}
\phi^2\mathbb{Q}-3(\phi.\mathbb{Q}.\phi)=0\,.
\end{equation}
\end{lemma}
\begin{proof}
A short computation gives
\[
\phi^2\mathbb{Q}-3(\phi.\mathbb{Q}.\phi)=\mathsf{Z}^i_{\phi }\Big(\phi^2\mathbb{C}_i-3(\mathbb{C}.\phi)\phi_i\Big)=0 \,,
\]
where $\mathbb{C}_i$ is the constraint \eqref{constraint_Dphi}, and hence vanishes.
\end{proof}

In order to solve \eqref{eeq2_components} we expand it as
\begin{equation}\label{equation_ALZ}
2\tensor{\varepsilon}{^i^r_s}\phi_i\mathbb{Z}_{lr}-\phi_l\mathbb{A}_s+\varepsilon_{lsr}\mathbb{L}^r=\mathbb{Q}_{ls}\,,
\end{equation}
where we have employed \eqref{other_notations}. By multiplying \eqref{equation_ALZ} by $\delta_{ls}$, $\phi_l$ and $\phi_s$, we respectively get
\begin{subequations}
\begin{align}
  & \mathbb{A}.\phi=-\frac{1}{3}\mathbb{Q}\,,\label{equation_A} \\
  & [\phi,\mathbb{L}]_s-\phi^2\mathbb{A}_s=(\phi.\mathbb{Q})_{s}\,,\label{equation_B} \\
  & [\phi,\mathbb{L}]_l-(\mathbb{A}.\phi)\phi_l=(\mathbb{Q}.\phi)_l\,.\phantom{\Big(}\label{equation_C}
\end{align}\end{subequations}
From these equations, we obtain

\begin{subequations}
\begin{align}
&\mathbb{A}_s\overset{\eqref{equation_B}}{=}\frac{1}{\phi^2}( [\phi,\mathbb{L}]_s-(\phi.\mathbb{Q})_{s})\underset{\eqref{equation_A}}{\overset{\eqref{equation_C}}{=}}\frac{1}{\phi^2}\left((\mathbb{Q}.\phi)_s-(\phi.\mathbb{Q})_{s}-\frac{1}{3}\mathbb{Q}\phi_s\right)\,,\label{equation_E}\\
&[\phi,\mathbb{L}]_l\underset{\eqref{equation_A}}{\overset{\eqref{equation_C}}{=}}(\mathbb{Q}.\phi)_l-\frac{1}{3}\mathbb{Q}\phi_l\label{equation_D}
\end{align}\end{subequations}
Here, and in the following, we assume $\phi^2(p)\neq0$ at every $p\in U$. Using now Lemma \ref{lemma_eq}, we solve \eqref{equation_D} for $\mathbb{L}_i$. The necessary condition for the solvability of this equation is \eqref{condition_Q}, and the solution
\begin{equation}\label{sol_L}
\mathbb{L}_i=-\frac{1}{\phi^2}[\phi,\mathbb{Q}.\phi]_i+N\phi_i\,,
\end{equation}
where $N$ is an arbitrary smooth function in $T^*Q$.

Once $\mathbb{A}_i$ and $\mathbb{L}_i$ have been obtained, we plug them back into \eqref{equation_ALZ}, which becomes
\[
[\phi,\mathbb{Z}_{l\cdot{}}]_s=\frac{1}{2}\mathbb{Q}_{ls}-\frac{1}{6\phi^2}\mathbb{Q}\phi_l\phi_s+\frac{1}{\phi^2}\phi_l(\mathbb{Q}.\phi)_s-\frac{1}{2\phi^2}\phi_l(\phi.\mathbb{Q})_s
-\frac{1}{2\phi^2}\phi_s(\mathbb{Q}.\phi)_l-\frac{N}{2}\varepsilon_{lsr}\phi^r\,.
\]
Again, this equation turns out to be of the type considered in Lemma \ref{lemma_eq} for each $l$ (an index which, since it is not contracted, can be effectively ignored). Hence, we can immediately solve it. Now, the necessary condition for its solvability is
\[
\phi_l\big(3(\phi.\mathbb{Q}.\phi)-\phi^2\mathbb{Q}\big)=0 \,,
\]
which, as we are assuming $\phi^2\neq0$ everywhere on $\Sigma$, is equivalent to \eqref{condition_Q}. The solution is
\begin{align}\label{Zli}
\begin{split}
\mathbb{Z}_{li}=&-\frac{1}{2\phi^2}\tensor{\varepsilon}{_i^j^k}\phi_j\mathbb{Q}_{lk}-\frac{1}{(\phi^2)^2}\phi_l[\phi,\mathbb{Q}.\phi]_i+\frac{1}{2(\phi^2)^2}\phi_l[\phi,\phi.\mathbb{Q}]_i\\
&+\frac{1}{2\phi^2}N(\phi_l\phi_i-\delta_{li}\phi^2)+\widetilde{N}_l\phi_i\,,
\end{split}
\end{align}
where $\widetilde{N}_l$, $l=1,2,3$ are three arbitrary smooth functions on $T^*Q$. 

At this point, we have arrived at \eqref{Zli} as a chain of implications; this means that the solutions to the original equation \eqref{eeq2_components} must have this form, but we have to check if after plugging \eqref{Zli} into \eqref{eeq2_components} we get conditions on the arbitrary functions $N$ and $\widetilde{N}_l$. A straightforward computation tells us that
\[
\widetilde{N}_l=-\frac{1}{(\phi^2)^2}[\phi,\phi.\mathbb{Q}]_l+\frac{1}{\phi^2}N \phi_l \,,
\]
with $N$ still arbitrary. The final solution to equation \eqref{eeq2_components} is then
\begin{align}\label{Zli2}
\begin{split}
\mathbb{Z}_{li}=&-\frac{1}{2\phi^2}\tensor{\varepsilon}{_i^j^k}\phi_j\mathbb{Q}_{lk}-\frac{1}{(\phi^2)^2}\phi_l[\phi,\mathbb{Q}.\phi]_i+\frac{1}{2(\phi^2)^2}\phi_l[\phi,\phi.\mathbb{Q}]_i\\
&-\frac{1}{(\phi^2)^2}\phi_i[\phi,\phi.\mathbb{Q}]_l+\frac{1}{2\phi^2}N(3\phi_l\phi_i-\delta_{li}\phi^2)\,.
\end{split}
\end{align}
This expression can be simplified with some effort by using the orthogonal decomposition 
\[
\mathsf{Z}_{\phi }=\frac{1}{\langle\phi\wed\phi\rangle}\big(\langle\phi\wed\mathsf{Z}_\phi\rangle\phi+[[\phi\wed \mathsf{Z}_\phi]\wed\phi]\big)
\]
The final result is given in the following Theorem, whose proof is provided by the preceding arguments.
\begin{theorem}\label{th_Z}
Let us assume that $\phi\in \Omega^0(\Sigma,\mathfrak{su}(2))$ is nowhere vanishing on $\Sigma$. Let us also assume that $(\mathrm{d}_A\phi)^i$ define a coframe on an open subset $U\subset\Sigma$, then the solution to the equation
\[{\tensor{\mathsf{P}}{_i_l^j^k}}\big(\mathsf{Z}^i_{\phi}F_k-(\mathrm{d}_A\phi)^i\wedge \mathsf{Z}_{Ak}\big)\phi_j=0\]
in the unknowns $\mathsf{Z}_A\in\Omega^1(U,\mathfrak{su}(2))$ and $\mathsf{Z}_\phi\in\Omega^0(U,\mathfrak{su}(2))$ on the subset of $T^*Q$ where the secondary constraints 
\[
\mathrm{d}_A[\phi\wed[F_A\wed\phi]]=0
\]
hold is given by $\mathsf{Z}_A=\mathsf{Z}^i_{A}(d_A\phi)_i$ with
\begin{equation}\label{solution_final}
\mathsf{Z}^i_{A}=\tensor{\varepsilon}{_l_j^k}\mathbb{F}^{il} \mathsf{Z}^j_{\phi}(\mathrm{d}_A\phi)_k+\mathsf{N}\big(3\phi^i\mathrm{d}\phi^2-2\phi^2(\mathrm{d}_A\phi)^i\big)
\end{equation}
where $\mathsf{Z}^j_{\phi}$ and $\mathsf{N}$ are arbitrary smooth functions  in $T^*Q$.
\end{theorem}

Several comments are in order now. 
\begin{remark}
    It is a simple exercise to show that \eqref{solution_final} satisfies \eqref{eeq1_components} and \eqref{eeq2_components} when the secondary constraints \eqref{constraint_Dphi} hold. 
\end{remark}
\begin{remark}
The Hamiltonian vector fields tangent to the subset of the phase space $T^*Q$ defined by the primary and secondary constraints depend on seven arbitrary functions in phase space: $a^i$, $\mathsf{Z}^i_{\phi}$, and $\mathsf{N}$. They have the form
\begin{align}\label{X}
\begin{split}
&\mathsf{X}_{\phi}=[\phi\wed a]+\mathsf{Z}_\phi\,,\\
&\mathsf{X}_A=\mathrm{d}_Aa+\mathsf{Z}_A\,,\\
&\mathsf{X}_{a}\,\quad \mathrm{arbitrary}\,.
\end{split}
\end{align}
with $\mathsf{Z}_A$ given by the expression appearing in Theorem \ref{th_Z}.     
\end{remark}
\begin{remark}
    On the basis of the Lie algebra that we are using, it is possible to show that, if we write $\mathsf{Z}^j_{\phi}=\xi\iprod (d_A\phi)^j$ with arbitrary $\xi \in \mathfrak{X}(\Sigma)$, then
    \[
    \tensor{\varepsilon}{_l_j^k}\mathbb{F}^{il}\mathsf{Z}^j_{\phi}(d_A\phi)_k=\pounds_\xi A^i-d_A(\xi\iprod A^i)\,,
    \]
and the final form of the Hamiltonian vector field is given by 
\begin{align}
\begin{split}
&\mathsf{X}_{\phi}=\pounds_\xi \phi +[\phi\wed (a-\xi \iprod A)]\,,\\
&\mathsf{X}_A=\pounds_\xi A+ \mathrm{d}_A (a-\xi \iprod A) + \mathsf{N} \Big( 3\mathrm{d}\langle \phi \wed  \phi \rangle \phi -2\langle \phi \wed \phi \rangle \mathrm{d}_A \phi \Big)\,,\\
&\mathsf{X}_{a}\,\quad \mathrm{arbitrary}\,.\phantom{\Big(}
\end{split}
\end{align} 
As we can see, we have the expected internal $SU(2)$ gauge transformations, with local gauge parameter $\Lambda:=a-\xi \iprod A$, and diffeomorphisms on $\Sigma$ generated by the vector field $\xi$. 
\begin{align*}
\delta_\xi(\phi,A):=(\pounds_\xi\phi,\pounds_\xi A)\,,\quad \delta_\Lambda(\phi,A):=([\phi\wed \Lambda],\mathrm{d}_A\Lambda )\,.
\end{align*}
In addition to these, we have an additional symmetry controlled by the $\mathsf{N}$-dependent term appearing in $\mathsf{X}_A$.
\end{remark}
\begin{remark}
By following Dirac's approach to the treatment of singular Hamiltonian systems, it is possible to find the linear combinations of primary constraints which are first-class and generate (via Poisson brackets) the 3-dimensional diffeomorphisms described by the Hamiltonian vector fields shown above. The remaining primary constraints are second-class. As we are following the GNH approach, we restrict ourselves to working on the primary constraint submanifold and, hence, we are not looking for these first-class constraints here. 
\end{remark}

%
%
\section{Comments and conclusions}{\label{sec_comments}}

In this paper, we have studied the Hamiltonian formulation of a recently proposed modification of the Husain-Kucha\v{r} model [that we refer to here as the Husain-Mehmood (HM) model] in order to interpret its local gauge symmetries and, in particular, to study how 3-dimensional diffeomorphisms show up. 

Within the theoretical physics community the method of choice to deal with such singular --constrained-- systems is the one developed by Dirac, frequently known as Dirac's ``algorithm''. This method was originally phrased in purely mechanical terms and its geometric meaning is, to this day, often put aside, although geometric interpretations are available \cite{BDMV}. Other similar --but not identical-- approaches to the same problem have been developed over the years (see, for instance, \cite{GNH1, GNH2, GNH3, gotay1979presymplectic}). These have a distinct geometric flavor and provide neat interpretations of the dynamics. From a conceptual point of view, they are powerful tools leading to clean Hamiltonian descriptions of singular field theories. They can be used, for instance, to obtain the Ashtekar formulation of GR in a particularly appealing way \cite{Barbero2024}. 

Very often Dirac's method is used as an intermediate step toward quantization. The key idea of Dirac's proposal to quantize constrained systems is to turn the so-called first-class constraints into operators in an appropriate Hilbert space and select the subspace given by the intersection of their kernels as the mathematical arena for physics (the \textit{physical Hilbert space}). If second-class constraints are present, the method can be adapted by introducing the so-called Dirac brackets.

An unwanted consequence of this way to proceed is that the details of the dynamics in Hamiltonian form, encoded in the Hamiltonian vector fields, are usually neglected as they are deemed irrelevant for quantization. Another frequent mistake is that the equations for the multipliers introduced in Dirac's procedure are almost never solved. This is very dangerous as it is known (and certainly recognized by Dirac himself \cite{Dirac, henneaux}) that the presence of arbitrary parameters in their solutions signals the existence of linear combinations of primary constraints that are first-class and, hence, generate gauge transformations. This is likely the reason why the 3-dimensional diffeos are not directly found in \cite{HM} but, instead, (incorrectly) interpreted in terms of internal $SU(2)$ transformations. A complete understanding of the nature of the constraints (in particular regarding their first or second-class character) is also needed to ``count'' the degrees of freedom as is usually done in physics.

It is important to point out that the equations for the multipliers appearing in Dirac's analysis are, in fact, identical to the equations for the Hamiltonian vector fields  \eqref{eeq1_components} and \eqref{eeq2_components}, so their solution cannot be sidelined. In our experience, most of the work needed to get a complete Hamiltonian description of a singular field theory, either in the Dirac or GNH approaches, goes into solving the equations for the Hamiltonian vector fields and checking the all-important tangency conditions. This is the reason why we have devoted a significant part of this paper to the resolution of the equations for the Hamiltonian vector fields. 

There are interesting generalizations of the present work. They would involve changing the gauge group \(G\) and/or considering non-trivial principal bundles. In the latter case the curvature \(F_A\) can be thought of as a section of the adjoint bundle \(\mathrm{ad} P\rightarrow M\), i.e. \(F_A\in \Omega^2(M,\mathrm{ad} P)\),  and consequently \(\Phi\in  \Omega^0(M,\mathrm{ad} P)\). For these models, the variational principle would be formally the same as \eqref{action_HM}, but some details of the Hamiltonian analysis will change.

\printbibliography

\end{document}